\documentclass[12pt,twoside]{article}
\usepackage{ocg}

\usepackage[latin1]{inputenc}
\usepackage[T1]{fontenc}

\usepackage{tikz}
\usetikzlibrary{automata,positioning,calc,through,fit,decorations.pathreplacing,matrix,decorations.markings}

\usepackage{xspace}
\usepackage{longtable}
\usepackage{array}
\usepackage{comment}

\usepackage{amsmath}
\usepackage{amssymb}
\usepackage{amsthm}

\usepackage{paralist}

\input{define}

\draftoff

\title{ON THE EXPRESSIVENESS OF PARIKH AUTOMATA AND RELATED MODELS}
\author{Micha\"el Cadilhac$^1$,
  Alain Finkel$^2$, and
  Pierre McKenzie$^1$}

\institute{$^1$~~DIRO,
  Universit\'e de Montr\'eal\\
  \texttt{\{cadilhac, mckenzie\}@iro.umontreal.ca}\\
  $^2$~~ LSV, CNRS \& \'Ecole Normale Sup\'erieure de Cachan\\
  \texttt{finkel@lsv.ens-cachan.fr}}

\def\revisionset! #1 !{\gdef\therevision{#1}}
\def\revision$LastChangedRevision:#1${\revisionset!#1!}
\revision$LastChangedRevision: 449 $

\begin{document}
\maketitle

\begin{abstract}
  \ifdraft
    \begin{tikzpicture}[remember picture,text=black!40,below,text
      centered,overlay]%
      \node (WD) [yshift=-1.5cm] at (current page.north)
      {\huge \bf\textsf{Working draft \today{}
          (rev. \therevision)}};
    \end{tikzpicture}%
    \fi The Parikh finite word automaton (PA) was introduced and studied by
    \KaR~\cite{klaedtke-ruess03}.  Natural variants of the PA arise from
    viewing a PA equivalently as an automaton that keeps a count of its
    transitions and semilinearly constrains their numbers.  Here we adopt
    this view and define the affine PA (APA), that extends the PA by having
    each transition induce an affine transformation on the PA registers, and
    the PA on letters (LPA), that restricts the PA by forcing any two
    transitions on same letter to affect the registers equally.  Then we
    report on the expressiveness, closure, and decidability properties of
    such PA variants.  We note that deterministic PA are strictly weaker than
    deterministic reversal-bounded counter machines.  We develop
    pumping-style lemmas and identify an explicit PA language recognized by
    no deterministic PA.  Our findings and the resulting overall picture are
    tabulated in our concluding section.
\end{abstract}
\thispagestyle{empty}


\section{Introduction}

Adding features to finite automata in order to capture situations beyond
regularity has been fruitful to many areas of research, in particular model
checking and complexity theory below NC$^2$ (e.g.,~\cite{klarlund89,str94}).
One such finite automaton extension is the \emph{Parikh automaton} (PA): A
PA~\cite{klaedtke-ruess03} is a pair $(A, C)$ where $C$ is a semilinear
subset of $\bbn^d$ and $A$ is a finite automaton over $(\Sigma \times D)$ for
$\Sigma$ a finite alphabet and $D$ a finite subset of $\bbn^d$. The PA
accepts the word $w_1\cdots w_n \in \Sigma^*$ if $A$ accepts a word $(w_1,
\bar{v_1})\cdots(w_n, \bar{v_n})$ such that $\sum \bar{v_i} \in C$.  \KaR
used PA to characterize an extension of (existential) monadic second-order
logic in which the cardinality of sets expressed by second-order variables is
available.

Here we carry the study of Parikh automata a little further.  First we
introduce related models of independent interest, each involving a finite
automaton $A$ and a \emph{constraint set} $C$ of vectors. (The main text has
formal definitions.)
\begin{inparaenum}[(1)]
\item \emph{Constrained automata} (CA) are defined to
  accept a word $w\in\Sigma^*$ iff the Parikh image of some accepting run of
  $A$ on $w$ 
  (i.e., the vector recording the number of occurrences of each transition
  along the run) belongs to $C$.
\item \emph{Affine Parikh automata} (APA) generalize PA by allowing each
  transition to perform a linear transformation on the $d$-tuple of PA
  registers prior to adding a new vector; an APA accepts a word $w$ iff some
  accepting run of $A$ on $w$ maps the all-zero vector to a vector in $C$.
\item \emph{Parikh automata on letters} (LPA) restrict PA by imposing the
  condition that any transition on $(a,\bar{u})\in(\Sigma\times D)$ and any
  transition on $(b,\bar{v})\in(\Sigma\times D)$ must satisfy
  $\bar{u} = \bar{v}$ when $a=b$.
\end{inparaenum}

Then our main observations are the following:
\begin{compactitem}
\item CA and deterministic CA respectively capture the class $\LPA$ of PA
  languages and the class $\LDPA$ of deterministic PA languages.
\item The language $\{a,b\}^*\cdot \{a^n\#a^n \st n \in \bbn\}$ belongs to
  $\LPA \setminus \LDPA$;
  these two classes were only proved different in~\cite{klaedtke-ruess03}.
\item APA and deterministic APA over $\bbq$ are no more powerful than the
  same models over $\bbn$.
\item APA express more languages than PA, and only context-sensitive
  languages; moreover the emptiness problem for deterministic APA is already
  undecidable.
\item Languages of LPA are equivalent to regular languages with a constraint
  on the Parikh image of their words.
\item Refining~\cite{klaedtke-ruess03} slightly, we compare our models with
  the reversal-bounded counter machines (RBCM) defined by
  Ibarra~\cite{ibarra78}, and show that $\LDPA$ is a strict subset of the
  languages expressed by deterministic RBCM.
\item Further expressiveness properties, closure properties, decidability
  properties and comparisons between the above models are derived.
  The overall resulting picture is summarized in
  tabular form in Section~\ref{sec:conclusion}.
\end{compactitem}


\section{Preliminaries}\label{sec:prelim}

We write $\bbz$ for the integers, $\bbn$ for the nonnegative integers,
$\bbn^+$ for $\bbn \setminus \{0\}$, $\bbq$ for the rational numbers, and
$\bbq^+$ for the strictly positive rational numbers.  We use $\bbk$ to denote
either $\bbn$ or $\bbq$.  Let $d, d' \in \bbn^+$.  Vectors in $\bbk^d$ are
noted with a bar on top, e.g., $\bar{v}$ whose elements are $v_1, \ldots,
v_d$.  For $C \subseteq \bbk^d$ and $D \subseteq \bbk^{d'}$, we write $C.D$
for the set of vectors in $\bbk^{d+d'}$ which are the concatenation of a
vector of $C$ and a vector of $D$.  We write $\bar{0} \in \{0\}^d$ for the
all-zero vector, and $\bar{e_i} \in \{0, 1\}^d$ for the vector having a $1$
only in position $i$.  We view $\bbk^d$ as the additive monoid $(\bbk^d, +)$.
For a monoid $(M, \cdot)$ and $S \subseteq M$, we write $S^*$ for the monoid
generated by $S$, i.e., the smallest submonoid of $(M, \cdot)$
containing~$S$.  A subset $E$ of $\bbk^d$ is \emph{$\bbk$-definable} if it is
expressible as a first order formula which uses the function symbols $+$,
$\lambda_e$ with $e \in \bbk$ corresponding to the scalar multiplication, and
the order $<$.  More precisely, a subset $E$ of $\bbk^d$ is $\bbk$-definable
iff there is such a formula with $d$ free variables, with $(x_1, \ldots, x_d)
\in E \Leftrightarrow \bbk \models \phi(x_1, \ldots, x_d)$.  Let us remark
that $\bbn$-definable sets are the Presburger-definable sets and they
coincide with the \emph{semilinear sets}~\cite{ginsburg-spanier66b}, i.e.,
finite unions of sets of the form $\{\bar{a_0} + k_1\bar{a_1} + \cdots +
k_n\bar{a_n} \st (\forall i)[k_i \in \bbn]\}$ for some $\bar{a_i}$'s in
$\bbn^d$.  Moreover, $\bbq$-definable sets are the semialgebraic sets defined
using affine functions\footnote{Semialgrebraic sets defined using affine
  functions are sometimes also called semilinear (e.g.,
  \cite{vandendries98}).  In this paper, we use ``semilinear'' only for
  $\bbn$-definable sets.}~\cite[Corollary I.7.8]{vandendries98}.

Let $\Sigma = \{a_1, \ldots, a_n\}$ be an (ordered) alphabet, and write
$\eps$ for the empty word.  The \emph{Parikh image} is the morphism
$\pkh\colon \Sigma^* \to \bbn^n$ defined by $\pkh(a_i) = \bar{e_i}$, for $1
\leq i \leq n$.  A language $L \subseteq \Sigma^*$ is said to be
\emph{semilinear} if $\pkh(L) = \{\pkh(w) \st w \in L\}$ is semilinear.  The
\emph{commutative closure} of a language $L$ is defined as the language $c(L)
= \{w \st \pkh(w) \in \pkh(L)\}$.  A language $L \subseteq \Sigma^*$ is said
to be \emph{bounded} if there exist $n > 0$ and $w_1, \ldots, w_n \in
\Sigma^+$ such that $L \subseteq w_1^*\cdots w_n^*$.  Two words $u, v \in
\Sigma^*$ are \emph{equivalent by the Nerode relation} (w.r.t.\ $L$), if for
all $w \in \Sigma^*$, $uw \in L \Leftrightarrow vw \in L$.  We then write $u
\equiv_L v$ (or $u \equiv v$ when $L$ is understood), and write $[u]_L$ for
the equivalence class of $u$ w.r.t.\ the Nerode relation.

We then fix our notation about automata.  An automaton is a quintuple $A =
(Q, \Sigma, \delta, q_0, F)$ where $Q$ is the finite set of states, $\Sigma$
is an alphabet, $\delta \subseteq Q \times \Sigma \times Q$ is the set of
transitions, $q_0 \in Q$ is the initial state and $F \subseteq Q$ are the
final states.  For a transition $t \in \delta$, where $t = (q, a, q')$, we
define $\orig(t) = q$ and $\dest(t) = q'$.  Moreover, we define $\mu_A\colon
\delta^* \to \Sigma^*$ to be the morphism defined by $\mu_A(t) = a$, and we
write $\mu$ when $A$ is clear from the context.  A \emph{path} on $A$ is a
word $\pi = t_1\cdots t_n \in \delta^*$ such that $\dest(t_i) =
\orig(t_{i+1})$ for $1 \leq i < n$; we extend $\orig$ and $\dest$ to paths,
letting $\orig(\pi) = \orig(t_1)$ and $\dest(\pi) = \dest(t_n)$.  We say that
$\mu(\pi)$ is the \emph{label} of $\pi$.  A path $\pi$ is said to be
\emph{accepting} if $\orig(\pi) = q_0$ and $\dest(\pi) \in F$; we let
$\apaths(A)$ be the language over $\delta$ of accepting paths on $A$.  We
then define $L(A)$, the \emph{language} of $A$, as the labels of the
accepting paths.



\section{Parikh automata}\label{sec:pa}

The following notations will be used in defining Parikh finite word automata
(PA) formally.  Let $\Sigma$ be an alphabet, $d \in \bbn^+$, and $D$ a finite
subset of $\bbn^d$.  Following~\cite{klaedtke-ruess03}, the monoid morphism
from $(\Sigma \times D)^*$ to $\Sigma^*$ defined by $(a, \bar{v}) \mapsto a$
is called the \emph{projection on $\Sigma$} and the monoid morphism from
$(\Sigma \times D)^*$ to $\bbn^d$ defined by $(a, \bar{v}) \mapsto \bar{v}$ is
called the \emph{extended Parikh image}.

\Remark Let $\Sigma = \{a_1, \ldots, a_n\}$ and $D \subseteq \bbn^n$.  If a
word $\omega \in (\Sigma \times D)^*$ is in $\{(a_i, \bar{e_i}) \st 1 \leq i
\leq n\}^*$, then the extended Parikh image of $\omega$ is the Parikh image
its projection on $\Sigma$.

\begin{definition}[Parikh automaton~\cite{klaedtke-ruess03}]
  Let $\Sigma$ be an alphabet, $d \in \bbn^+$, and $D$ a finite subset of
  $\bbn^d$.  A \emph{Parikh automaton (PA)} of dimension $d$ over $\Sigma
  \times D$ is a pair $(A, C)$ where $A$ is a finite automaton over $\Sigma
  \times D$, and $C \subseteq \bbn^d$ is a semilinear set.  The PA language,
  written $L(A, C)$, is the projection on $\Sigma$ of the words of $L(A)$
  whose extended Parikh image is in $C$.
  The PA is said to be \emph{deterministic (DetPA)} if for every state $q$ of
  $A$ and every $a \in \Sigma$, there exists at most one pair $(q', \bar{v})$
  with $q'$ a state and $\bar{v} \in D$ such that $(q, (a, \bar{v}), q')$ is
  a transition of $A$.  We write $\LPA$ (resp. $\LDPA$) for the class of
  languages recognized by PA (resp. DetPA).
\end{definition}

An alternative view of the PA will prove very useful.  Indeed we note that a
PA can be viewed equivalently as an automaton that applies a semilinear
constraint on the counts of the individual transitions occurring along its
accepting runs.  To explain this, let $(A, C)$ be a PA of dimension $d$, and
let $\delta = \{t_1, \ldots, t_n\}$ be the transitions of $A$.
Consider the automaton $B$ which is a copy of $A$ except that the vector part
of the transitions is dropped, and suppose there is a natural bijection
between the transitions of the two automata.  Let $\pi$ be a path in $A$; the
contribution to the extended Parikh image of $\mu(\pi)$ of the transition
$t_i = (p, (a, \bar{v_i}), q)$ is $\bar{v_i}$; thus, knowing how many times
$t_i$ appears in the path traced by $\pi$ in $B$ is enough to retrieve the
value of the extended Parikh image of $\mu(\pi)$.
Now note that the bijection exists if no two distinct transitions $t_i, t_j$
are such that $t_i = (p, (a, \bar{v_i}), q)$ and $t_j = (p, (a, \bar{v_j}),
q)$.  However, if such $t_i$ and $t_j$ exist, we can replace them by $t = (p,
(a, \bar{e_{d+1}}), q)$, incrementing in the process the dimension of PA, and
change $C$ to $C'$ defined by $(\bar{v}, c) \in C' \Leftrightarrow (\exists
c_i)(\exists c_j)[c = c_i + c_j \land \bar{v} + c_i.\bar{v_i} + c_j.\bar{v_j}
\in C]$ without changing the language of the PA.  It is thus readily seen
that the following defines models equivalent to the PA%
\footnote{Another equivalent view of PA languages suggested by one referee is as sets
  $R^{-1}(X)$ where $R$ is a rational relation over $\Sigma^*\times \bbn^d$
  and $X$ is a rational subset of $\bbn^d$.  An artificial further restriction to this viewpoint would serve to
  capture DetPA languages.}  and the DetPA:\squeezetopsep
\begin{definition}[Constrained automaton]
  A \emph{constrained automaton (CA)} over an alphabet~$\Sigma$ is a pair
  $(A, C)$ where $A$ is a finite automaton over $\Sigma$ with $d$
  transitions, and $C \subseteq \bbn^d$ is a semilinear set.  Its language is
  $L(A, C) = \{\mu(\pi) \st \pi \in \apaths(A) \land \pkh(\pi) \in C\}$.
  The CA is said to be \emph{deterministic (DetCA)} if $A$ is deterministic.
\end{definition}



\subsection{On the expressiveness of Parikh automata}

The constrained
automaton characterization of PA helps deriving pumping-style
necessary conditions for membership in $\LPA$ and in $\LDPA$:
\squeezetopsep
\begin{lemma}\label{lem:tdbpa}%
  Let $L \in \LPA$.  There exist $p, \ell \in \bbn^+$ such that any $w \in L$
  with $|w| > \ell$ can be written as $w = uvxvz$ where:
  \begin{compactenum}
  \item $0 < |v| \leq p$, $|x| > p$, and $|uvxv| \leq \ell$,
  \item $uv^2xz \in L$ and $uxv^2z \in L$.
  \end{compactenum}
\end{lemma}
\begin{proof}
  Let $(A, C)$ be a CA of language $L$.  Let $p$ be the number of states in
  $A$ and $m$ be the number of elementary cycles (i.e., cycles in which no
  state except the start state occurs twice) in the underlying multigraph of
  $A$.  Finally, let $\ell = p \times (2m+1)$.  Now, let $w \in L$ such that
  $|w| \geq \ell$ and $\pi \in \apaths(A)$ such that $\mu(\pi) = w$ and
  $\pkh(\pi) \in C$.  Write $\pi$ as $\pi_1\cdots\pi_{2m+1}\rho$ where
  $|\pi_i| = p$.  By the pigeonhole principle, each $\pi_i$ contains an
  elementary cycle, and thus, there exist $1 \leq i , j \leq m+1$ with $i+1 <
  j$ such that $\pi_i$ and $\pi_j$ share the same cycle $\eta_v$ labeled with
  a word $v$.  Write:
  \begin{compactitem}
  \item $\pi_i$ as $\pi_{i,1}\eta_v\pi_{i,2}$, and $\pi_j$ as $\pi_{j,1}\eta_v\pi_{j,2}$, 
  \item $\eta_u$ for $\pi_1\cdots \pi_{i-1}\pi_{i,1}$ and $u$ for $\mu(\eta_u)$,
  \item $\eta_x$ for $\pi_{i,2}\pi_{i+1}\cdots \pi_{j-1}\pi_{j,1}$ and $x$
    for $\mu(\eta_x)$,
  \item $\eta_z$ for $\pi_{j,2}\pi_{j+1}\cdots\pi_{\ell+1}\rho$ and $z$ for $\mu(\eta_z)$.
  \end{compactitem}
  \vspace{-\baselineskip}Then $\pi = \eta_u\eta_v\eta_x\eta_v\eta_z$ and $w =
  uvxvz$.  Moreover, both $\pi' = \eta_u\eta_v^2\eta_x\eta_z$ and $\pi'' =
  \eta_u\eta_x\eta_v^2\eta_z$ are accepting paths with the same Parikh image
  as $\pi$.  Thus, $\mu(\pi')=uv^2xz \in L$ and $\mu(\pi'')=uxv^2z \in L$.
  Moreover, $0 <|v|\leq p, |x|>p$ and $|uvxv|\leq \ell$.
\end{proof}

A similar argument leads to a stronger property for the languages
  belonging to $\LDPA$:
\begin{lemma}\label{lem:tdbdpa}%
  Let $L \in \LDPA$.  There exist $p, \ell \in \bbn^+$ such that any $w$ over the
  alphabet of $L$ with $|w| > \ell$ can be written as $w = uvxvz$ where:
  \begin{compactenum}
  \item $0 < |v| \leq p$, $|x| > p$ and $|uvxv| \leq \ell$,
  \item $uv^2x, uvxv$ and $uxv^2$ are equivalent w.r.t.\ the Nerode
    relation of $L$.
  \end{compactenum}
\end{lemma}

\pagebreak
We apply Lemma~\ref{lem:tdbpa}
to the language $\COPY$, defined as $\{w\#w
\st w \in \{a, b\}^*\}$, as follows:\squeezetopsep
\begin{proposition}\label{prop:copy}%
  $\COPY \not\in \LPA$.
\end{proposition}
\begin{proof}
  Suppose $\COPY \in \LPA$.  Let $\ell, p$ be given by Lemma~\ref{lem:tdbpa},
  and consider $w = (a^pb)^\ell\#(a^pb)^\ell \in \COPY$.  Lemma~\ref{lem:tdbpa}
  states that $w=uvxvz$ where $uvxv$ lays in the first half of $w$, and $s =
  uv^2xz \in \COPY$.  Note that $x$ contains at least one $b$.  Suppose $v =
  a^i$ for $1 \leq i \leq p$, then there is a sequence of $a$'s in the first
  half of $s$ unmatched in the second half.  Likewise, if $v$ contains a $b$,
  then $s$ has a sequence of $a$'s between two $b$'s unmatched in the second
  half.  Thus $s \not\in \COPY$, a contradiction.  Hence $\COPY \not\in \LPA$.
\end{proof}

As \KaR show using closure properties, DetPA are strictly weaker than PA.
The thinner grain of Lemma~\ref{lem:tdbdpa} suggests explicit languages that
witness the separation of $\LDPA$ from $\LPA$.  Indeed, let $\EQUAL \subseteq
\{a,b,\#\}^*$ be the language $\{a,b\}^*\cdot \{a^n\#a^n \st n \in \bbn\}$,
we have:\squeezetopsep
\begin{proposition}\label{prop:equal}
  $\EQUAL \in \LPA \setminus \LDPA$.
\end{proposition}
\begin{proof}
  We omit the proof that $\EQUAL \in \LPA$.  Now, suppose $\EQUAL \in \LDPA$,
  and let $\ell, p$ be given by Lemma~\ref{lem:tdbdpa}.  Consider $w =
  (a^pb)^\ell$. Lemma~\ref{lem:tdbdpa} then asserts that a prefix of $w$ can
  be written as $w_1 = uvxv$, and that $w_2 = uv^2x$ verifies $w_1 \equiv
  w_2$.  As $|x| > p$, $x$ contains a $b$.  Let $k$ be the number of $a$'s at
  the end of $w_1$.  Suppose $v = a^i$ for $1 \leq i \leq p$, then $w_2$ ends
  with $k - i < k$ letters $a$.  Thus $w_1\#a^k \in \EQUAL$ and $w_2\#a^k
  \not\in \EQUAL$, a contradiction.  Suppose then that $v = a^iba^k$, with $0
  \leq i+k <p$.  Then $w_2$ ends with $p-i>k$ letters $a$, and similarly,
  $w_1 \not\equiv w_2$, a contradiction.  Thus $\EQUAL \not\in \LDPA$.
\end{proof}

For comparison, we mention another line of attack for the study of $\LDPA$.
The proof is omitted, but is based on the number of possible configurations
of a PA, which is polynomial in the length of the input word.  \KaR used a
similar argument to show that $\PAL = \{w\#w^R \st w\in \{a,b\}^+\}$, where
$w^R$ is the reversal of $w$, is not in $\LPA$.\squeezetopsep
\begin{lemma}\label{lem:combi}
  Let $L \in \LDPA$.  Then there exists $c > 0$ such that $|\{[w]_L \st w \in
  \Sigma^n\}| \in O(n^c)$.
\end{lemma}
\squeezeafterenv
\begin{proposition}
  Let $L = \{w \in \{a, b\}^* \st w_{|w|_a} = b\}$, where $w_i$ is the $i$-th
  letter of $w$.  Then $L \in \LPA \setminus \LDPA$.
\end{proposition}
\begin{proof}
  \def\halfn{{n \over 2}}%
  We omit the proof that $L \in \LPA$; the main point is simply to guess the
  position of the $b$ referenced by $|w|_a$.  On the other hand, let $n > 0$
  and $u, v \in \{a, b\}^n$ such that $|u|_a = |v|_a = \halfn$ and there
  exists $p \in \{\halfn, \ldots, n\}$ with $u_p \neq v_p$.  Let $w =
  a^{p-\halfn}$, then $(uw)_{|uw|_a} = (uw)_{|u|_a+|w|_a} = (uw)_p = u_p$,
  and similarly, $(vw)_{|vw|_a} = v_p$.  This implies $uw \not\in L
  \leftrightarrow vw \in L$, thus $u \not\equiv v$.  Then for $0 \leq i \leq
  \halfn$, define $E_i = \{a^{\halfn - i}b^iz \st z \in \{a, b\}^\halfn \land
  |z|_a = i\}$.  For any $u,v \in \bigcup E_i$ with $u \neq v$, the previous
  discussion shows that $u \not\equiv v$.  Thus $|\{[w]_L \st w \in
  \{a,b\}^n\}| \geq |\bigcup_{i=0}^\halfn E_i| = \sum_{i=0}^\halfn
  |E_i| = \sum_{i = 0}^\halfn {\halfn \choose i} = 2^\halfn \not\in
  O(n^{O(1)})$.  Lemma~\ref{lem:combi} then implies that $L \not\in \LDPA$.
\end{proof}


\subsection{On decidability and closure properties of Parikh automata}

The following table summarizes decidability results for PA and DetPA.  The
results in bold are new, while the others are from \cite{klaedtke-ruess03}
and \cite{ibarra78}:\squeezetopsep
\begin{center}
  \setlength{\tabcolsep}{10pt}
  \addtolength{\extrarowheight}{2pt}
  \begin{tabular}{c|c|c|c|c|c|}
    \cline{2-6}
    & $= \emptyset$   & $= \Sigma^*$   & is finite  & $\subseteq$ & is regular\\
    \hline
    \multicolumn{1}{|c|}{DetPA}   & D & D & {\bf D}  & {\bf D} & ?\\
    \hline
    \multicolumn{1}{|c|}{PA}  & D & U & {\bf D} & {\bf U} & {\bf U}\\
    \hline
  \end{tabular}
\end{center}

\begin{proposition}
  (1) Finiteness is decidable for PA.  (2) Inclusion is decidable for DetPA and
  undecidable for PA. (3) Regularity is undecidable for PA.
\end{proposition}
\begin{proof}
  \emph{(1).} Let $(A, C)$ be a CA.  Then $\apaths(A)$ is a regular language,
  and thus, its Parikh image is effectively semilinear (this is a special
  case of Parikh's theorem~\cite{parikh66}).  It follows that the language
  described by $A$ and $C$ is finite if and only if $\pkh(\apaths(A)) \cap C$
  is finite, which is decidable. ~~\emph{(2).}  Decidability of inclusion
  for DetPA follows from the fact that $\LDPA$ is closed under complement and
  intersection, and that the emptiness problem is decidable for DetPA.  (In
  fact, it is decidable whether the language of a PA is included in the
  language of a DetPA.)  Undecidability of inclusion for PA follows
  immediately from the undecidability of the universe problem for PA.
  ~~\emph{(3).} This follows from a theorem of~\cite{greibach68}, which
  states the following: Let $\calC$ be a class of languages closed under
  union and under concatenation with regular languages.  Let $P$ be a
  predicate on languages true of every regular language, false of some
  languages, preserved by inverse rational transduction, union with
  $\{\eps\}$ and intersection with regular languages.  Then $P$ is
  undecidable in $\calC$.  Obviously, $\LPA$ satisfies the hypothesis for
  $\calC$.  Moreover, ``being regular in $\LPA$'' is a predicate satisfying
  the hypothesis for $P$.  Thus, regularity is undecidable for PA.
\end{proof}

We now further the study of closure properties of PA and DetPA started
in~\cite{klaedtke-ruess03}.  The following table collects the closure
properties of PA and DetPA, where $h$ is a morphism, $c$ is the commutative
closure.
In bold are the results of the present paper, while the other results can be
found in~\cite{klaedtke-ruess03} (detailed proofs by Karianto can be found
in~\cite{karianto04}):
\begin{center}
  \setlength{\tabcolsep}{10pt}
  \addtolength{\extrarowheight}{2pt}
  \begin{tabular}{c|c|c|c|c|c|c|c|c|}
    \cline{2-9}
    & $\cup$  & $\cap$ & $\cdot$ & $\bar{\phantom{\cdot\;}}$ & $h$   & $h^{-1}$ & $c$ & ${}^*$\\
    \hline
    \multicolumn{1}{|c|}{DetPA}   & Y & Y & {\bf N} & Y & {\bf N} & Y & {\bf Y} &  {\bf N}  \\
    \hline
    \multicolumn{1}{|c|}{PA}    & Y & Y & Y & N & Y & Y  & {\bf Y}   & N  \\
    \hline
  \end{tabular}
\end{center}

As the language $\EQUAL$ separating $\LDPA$ from $\LPA$ is the concatenation
of a regular language and a language of $\LDPA$, we have:\squeezetopsep
\begin{proposition}%
  $\LDPA$ is not closed under concatenation.
\end{proposition}

\begin{proposition}\label{prop:clotpa}
  (1) The commutative closure of any semilinear language is in $\LDPA$.
  (2)
  $\LDPA$ is not closed under morphisms.
\end{proposition}
\begin{proof}
  \emph{(1).} Let $\Sigma = \{a_1, \ldots, a_n\}$, $L \subseteq \Sigma^*$ a
  semilinear language, and $C = \pkh(L)$.  Define $A$ to be an automaton with
  one state, initial and final, with $n$ loops, the $i$-th labeled $(a_i,
  \bar{e_i}) \in \Sigma \times \{\bar{e_i}\}_{1 \leq i \leq n}$.  Then $c(L)
  = L(A, C)$.  \emph{(2)} is straightforward
  as any language of $\LPA$ is the image by a morphism of a language in
  $\LDPA$.  Indeed, say $(A, C)$ is a CA and let $B$ be the copy of $A$ in
  which the transition $t$ is relabeled $t$; then $B$ is deterministic and
  $L(A, C) = \mu_A(L(B, C))$.  This implies the nonclosure of $\LDPA$ under
  morphisms.
  %
\end{proof}

Note that \emph{(1)} from Proposition~\ref{prop:clotpa} implies that
both $\LPA$ and $\LDPA$ are closed under commutative closure, as both are
classes of semilinear languages~\cite{klaedtke-ruess03}.

\begin{proposition}
  Neither $\LPA$ nor $\LDPA$ is closed under starring.
\end{proposition}
\begin{proof}
  We show that the starring of $L = \{a^nb^n \st n \in \bbn\}$ is not in
  $\LPA$.  Suppose $L^* \in \LPA$, and let $w = (a^pb^p)^\ell$, where
  $\ell,p$ are given by Lemma~\ref{lem:tdbpa}.  The same lemma asserts that
  $w = uvxvz$, such that, in particular, $uv^2xz$ and $uxv^2z$ are in $L^*$.
  Now suppose $v = a^i$ for some $i \leq p$.  Then $uv^2x$ contains
  $a^{p+i}b^p$ with no more $b$'s on the right.  Thus $uv^2xz \not\in L^*$.
  The case for $v = b^i$ is similar.  Now suppose $v = a^ib^j$ with $i,j >
  0$.  Then $uv^2x$ contains $\cdots a^pb^ja^ib^p \cdots$, but $i < p$, thus
  $uv^2xz \not\in L^*$.  The case $v = b^ia^j$ is similar.  Thus $L^* \not\in
  \LPA$.
\end{proof}
\squeezetopsep%
\Remark Baker and Book~\cite{baker-book74} already note, in different terms,
that if $\LPA$ were closed under starring, it would be an intersection closed
full AFL containing $\{a^nb^n \st n \geq 0\}$, and so would be equal to the
class of Turing-recognizable languages.  Thus $\LPA$ is not closed under
starring.


\subsection{Parikh automata and reversal-bounded counter machines}

\KaR noticed in~\cite{klaedtke-ruess02} that Parikh automata recognize the
same languages as reversal-bounded counter machines, a model introduced by
Ibarra~\cite{ibarra78}:
\begin{definition}[Reversal-bounded counter machine~\cite{ibarra78}]
  A \emph{one-way, $k$-counter machine} $M$ is a 5-uple $(Q,
  \Sigma, \delta, q_0, F)$ where $Q$ is a finite set of states, $\Sigma$ is
  an alphabet, $\delta \subseteq Q\times (\Sigma \cup \{\sharp\}) \times \{0,
  1\}^k \times Q \times \{S, R\} \times \{-1, 0, +1\}^k$ is the transition
  function, $q_0 \in Q$ is the initial state and $F \subseteq Q$ is the set
  of final states.  Moreover, we suppose $\sharp \not\in \Sigma$.  The
  machine is \emph{deterministic} if for any $(p, \ell, \bar{x})$, there
  exists at most one $(q, h, \bar{v})$ such that
  $(p,\ell,\bar{x},q,h,\bar{v}) \in \delta$.
  On input $w$, the machine starts with a read-only tape containing
  $w\sharp$, and its head on the first character of~$w$.  Write $c_i$ for the
  $i$-th counter, then a transition $(p,\ell,\bar{x},q,h,\bar{v}) \in \delta$
  is taken if the machine is in state $p$, reading character $\ell$ and $c_i
  = 0$ if $x_i = 0$ and $c_i > 0$ if $x_i = 1$, for all $i$.  The machine
  then enters state $q$, its head is moved to the right iff $h = R$, and
  $\bar{v}$ is added to the counters.  If the head falls off the tape, or if
  a counter turns negative, the machine rejects.  A word is accepted if an
  execution leads to a final state.
  The machine is \emph{reversal-bounded (RBCM)} if there exists an integer
  $r$ such that any accepting run changes between increments and decrements
  of the counters a (bounded) number of times less than $r$.  We write
  \emph{DetRBCM} for deterministic RBCM.  We write $\LRBCM$ (resp. $\LDRBCM$)
  for the class of languages recognized by RBCM (resp. DetRBCM).
\end{definition}\squeezeafterenv

In~\cite[Section A.3]{klaedtke-ruess02}, it is shown that PA
have the same expressive power as (nondeterministic) RBCM.  Although Fact~30
of~\cite{klaedtke-ruess02}, on which the authors rely to prove this result,
is technically false as stated,\footnote{Fact~30 of~\cite{klaedtke-ruess02}
  states the following.  Consider a RBCM $M$ which, for any counter, changes
  between increment and decrement only once.  Let $M'$ be $M$ in which
  negative counter values are allowed and the zero-tests are ignored.  Then a
  word is claimed to be accepted by $M$ iff the run of $M'$ on the same word
  reaches a final state with all its counters nonnegative.  A counter-example
  is the following.  Take $A$ to be the minimal automaton for $a^*b$, and add
  a counter for the number of $a$'s that blocks the transition labeled $b$
  unless the counter is nonzero.  This machine recognizes $a^+b$.  Then by
  removing this test, the machine now accepts $b$.} the small gap there can
be fixed so that:\squeezetopsep
\begin{proposition}[\cite{klaedtke-ruess02}]\label{prop:lparbcm}%
  $\LPA = \LRBCM$.
\end{proposition}
\pagebreak
\def\scirc{\hbox{\footnotesize $\spadesuit$}}\def\sbullet{\hbox{\footnotesize
    $\clubsuit$}}\def\sdiamond{\hbox{\footnotesize $\heartsuit$}}%
Further, we study how the notion of determinism compares in
the two models.  Let $\NSUM = \{a^n\scirc b^{m_1}\#b^{m_2}\#\cdots
\#b^{m_k}\sbullet c^{m_1+\cdots+m_n}\st k \geq n \geq 0 \land (\forall i)[m_i
\in \bbn]\}$: the number of $a$'s is the number of $m_i$'s to add to get the
number of $c$'s.  Note that $\NSUM$ is not context-free.  Then:\squeezetopsep
\begin{proposition}\label{prop:ldpadrbcm}%
  $\LDPA \subsetneq \LDRBCM$ and $\NSUM \in \LDRBCM \setminus \LDPA$.
\end{proposition}
\begin{proof}
  We first show that $\LDPA \subseteq \LDRBCM$.  Let $(A, C)$ be a CA, where
  $A = (Q, \Sigma, \delta, q_0, F)$ is deterministic and let $\delta = \{t_1,
  \ldots, t_k\}$.  We define a DetRBCM of the same language in two steps.
  (1) First, let $M$ be the $k$-counter machine $(Q \cup \{q_f\}, \Sigma,
  \zeta, q_0, q_f)$, where $q_f \not\in Q$ and $\zeta$ is defined
  by:
  $$\zeta = \bigcup_{\bar{x} \in \{0,1\}^k} \bigg(\big\{(q, a, \bar{x}, q', R,
  \bar{e_i}) \st t_i = (q, a, q')\big\} \cup \big\{(q, \sharp, \bar{x}, q_f,
  S, \bar{0}) \st q \in F\big\}\bigg).$$%
  This machine (trivially a DetRBCM) does not make any test, and accepts (in
  $q_f$) precisely the words accepted by $A$.  Moreover, the state of the
  counters in $q_f$ is the Parikh image of the path taken (in $A$) to
  recognize the input word.  (2) We then refine $M$ to check that the counter
  values belong to $C$.
  We note that we can do that as a direct consequence of the proof
  of~\cite[Theorem 3.5]{ibarra-su99}, but this proof relied on nontrivial
  algebraic properties of systems $A\bar{y} = \bar{b}$, where $A$ is a
  matrix, $\bar{y}$ are unknowns and $\bar{b}$ is a vector; we present here
  an elementary proof.
  Recall that $C$ can be expressed as a quantifier-free first-order formula
  which uses the function symbol $+$, the congruence relations $\equiv_i$,
  for $i \geq 2$, and the order relation $<$ (see, e.g.,~\cite{enderton72}).
  So let $C$ be given as such formula $\phi_C$ with $k$ free variables.  Let
  $\phi_C$ be put in disjunctive normal form.  The machine $M$ then tries
  each and every clause of $\phi_C$ for acceptance.  First, note that a term
  can be computed with a number of counters and reversals which depends only
  on its size: for instance, computing $c_i + c_j$ requires two new counters
  $x, y$; $c_i$ is decremented until it reaches $0$, while $x$ and $y$ are
  incremented, so that their value is $c_i$; now decrement $y$ until it
  reaches $0$ while incrementing $c_i$ back to its original value; then do
  the same process with $c_j$: as a result, $x$ is now $c_i+c_j$.  Second,
  note that any atomic formula ($t_1 < t_2$ or $t_1 \equiv_i t_2$) can be
  checked by a DetRBCM: for $t_1 < t_2$, compute $x_1 = t_1$ and $x_2 = t_2$,
  then decrement $x_1$ and $x_2$ until one of them reaches $0$, if the first
  one is $x_1$, then the atomic formula is true, and false otherwise; for
  $t_1 \equiv_i t_2$, a simple automaton-based construction depending on $i$
  can decide if the atomic formula is true.  Thus, a DetRBCM can decide, for
  each clause, if all of its atomic formulas (or negation) are true, and in
  this case, accept the word.  This process does not use the read-only head,
  and uses a number of counters and a number of reversals bounded by the
  length of $\phi_C$.

  We now show that $\NSUM \in \LDRBCM \setminus \LDPA$.  We omit the fact
  that $\NSUM \in \LDRBCM$.
  %
  Now suppose $(A, C)$ is a DetPA such that $L(A, C) = \NSUM$, with $A = (Q,
  \Sigma \times D, \delta, q_0, F)$ also deterministic.  We may suppose that
  the projection on $\Sigma$ of $L(A)$ is a subset of
  $a^*\scirc(b^*\#)^*b^*\sbullet c^*$, so that there exist $k \geq 0$, $q_1,
  \ldots, q_k \in Q$, and $j \in \{0, \ldots, k\}$ such that $(q_i, (a,
  \bar{v_i}), q_{i+1}) \in \delta$, for $0 \leq i < k$ and some
  $\bar{v_i}$'s, and $(q_k, (a, \bar{v_k}), q_j) \in \delta$.  Moreover, we
  may suppose that no other transition points to one of the $q_i$'s, and that
  all transitions $t = (q_i, (\ell, \bar{v}), q) \in \delta$ such that $q
  \not\in \{q_0, \ldots, q_k\}$ are with $\ell = \scirc$; let $T$ be the set
  of all such transitions~$t$.  We define $|T|$ DetPA such that the union of
  their languages is $\SUMN = \{\scirc w\sdiamond a^n \st a^n\scirc w \in
  \NSUM\}$, that is, the strings of $\NSUM$ with $a^n$ pushed at the end.
  For $t \in T$, define $A_t$ as the automaton similar to $A$ but which
  starts with the transition $t$ and delay the first part of the computation
  until the very end.  Formally, $A_t = (Q \cup \{q_0'\}, \Sigma \times D,
  \delta_t, q_0', \{\orig(t)\})$ where $\delta_t = (\delta \setminus T) \cup
  \{(q_0', \mu(t), \dest(t)\} \cup \{(q_f, (\sdiamond, \bar{0}), q_0) \st q_f
  \in F\}$ with $q_0'$ a fresh state.  Now for $\omega \in L(A)$, let $t$ be
  the transition labeled $\scirc$ taken when $A$ reads $\omega$, and let
  $\omega = \omega_1\mu(t)\omega_2$.  Then $\mu(t)\omega_2(\sdiamond,
  \bar{0})\omega_1 \in L(A_t)$, and this word has the same extended Parikh
  image as $\omega$.  Thus we have that $\bigcup_{t \in T} L(A_t, C) =
  \SUMN$, and if $\NSUM \in \LDPA$, then $\SUMN \in \LDPA$.  A proof similar
  to Proposition~\ref{prop:equal} then shows that $\SUMN \not\in \LDPA$, a
  contradiction; thus $\NSUM \not\in \LDPA$.
\end{proof}

The parallel drawn between (Det)PA and (Det)RBCM allows transferring some
RBCM and DetRBCM results to PA and DetPA.  An example is a consequence of the
following lemma proved in 2011 by Chiniforooshan \emph{et al}.\
\cite{chiniforooshan-daley-ibarra-kari-seki11} for the purpose of showing
incomparability results between different models of reversal-bounded counter
machines:\squeezetopsep
\begin{lemma}[\cite{chiniforooshan-daley-ibarra-kari-seki11}]\label{lem:chini}
  Let a DetRBCM express $L\subseteq\Sigma^*$.  Then there exists
  $w\in\Sigma^*$ such that $L \cap w\Sigma^*$ is a nontrivial regular
  language.
\end{lemma}
\squeezeafterenv Variants of the language $\EQUAL$ from
Proposition~\ref{prop:equal} can be shown outside $\LDPA$ in this way.  For
instance, for $\Sigma=\{a,b\}$, $\SANBN = \Sigma^* \cdot \{a^nb^n \st n \in
\bbn\}$ is such that any $w \in \Sigma^*$ makes $\SANBN \cap w\Sigma^*$
nonregular.  Although Lemma~\ref{lem:chini} thus gives languages in $\LPA
\setminus \LDPA$, Lemma~\ref{lem:chini} seemingly does not apply to $\EQUAL$
itself since $\EQUAL \cap \#\{a,b,\#\}^* = \{\#\}$ is regular.


\section{Affine Parikh automata}\label{sec:apa}

A PA of dimension $d$ can be viewed as an automaton in which each transition
updates a vector $\bar{x}$ of $\bbn^d$ using a function $\bar{x} \leftarrow
\bar{x} + \bar{v}$ where $\bar{v}$ depends only on the transition.  At the
end of an accepting computation, the word is accepted if $\bar{x}$ belongs to
some semilinear set.  We propose to generalize the updating function to an
affine function.  We start by defining the model, and show that defining it
over $\bbn$ is at least as general as defining it on $\bbq$.  We study the
expressiveness of this model, and show it is strictly more powerful than PA.
We then note that deterministic such automata can be normalized so as to
essentially trivialize their automaton component.  We then study nonclosure
properties and decidability problems associated with APA, leading to the
observation that APA lack some desirable properties --- e.g., properties
usually needed for any real-world application.

In the following, we consider the vectors in $\bbk^d$ to be \emph{column}
vectors.  Let $d, d' > 0$.  A function $f\colon \bbk^d \to \bbk^{d'}$ is a
(total) \emph{affine function} if there exist a matrix $M \in \bbk^{d'\times
  d}$ and $\bar{v} \in \bbk^{d'}$ such that for any $\bar{x} \in \bbk^d$,
$f(\bar{x}) = M.\bar{x} + \bar{v}$; it is \emph{linear} if $\bar{v} =
\bar{0}$.  We note such a function $f = (M, \bar{v})$.  We write
$\func_d^\bbk$ for the set of affine functions from $\bbk^d$ to $\bbk^d$ and
view $\func_d^\bbk$ as the monoid $(\func_d^\bbk, \diamond)$ with $(f
\diamond g) (\bar{x}) = g(f(\bar{x}))$.

\begin{definition}[Affine Parikh automaton]
  A $\bbk$-\emph{affine Parikh automaton ($\bbk$-APA)} of dimension $d$ is a
  triple $(A, U, C)$ where $A$ is an automaton with transition set $\delta$,
  $U$ is a morphism from $\delta^*$ to $\func_d^\bbk$ and $C \subseteq
  \bbk^d$ is a $\bbk$-definable set; recall that $U$ need only be defined on
  $\delta$.  The language of the APA is $L(A, U, C) = \{\mu(\pi) \st \pi \in
  \apaths(A) \land (U(\pi))(\bar{0}) \in C\}$.  The \KAPA is said to be
  \emph{deterministic (\KDAPA)} if $A$ is.  We write \emph{$\LKAPA$}
  (resp. $\LKDAPA$) for the class of languages recognized by \KAPA
  (resp. \KDAPA).
\end{definition}
\squeezeafterenv
\Remark It is easily seen that \NAPA (resp. \NDAPA) are a generalization of
CA (resp. DetCA).  Indeed, let $(A, C)$ be a CA, and let $\pkh$ be the Parikh
image over the set $\delta$ of transitions of $A$.  Define, for $t\in\delta$,
$U(t) = (\Id, \pkh(t))$ where $\Id$ is the identity matrix of dimension
$|\delta| \times |\delta|$.  Then $L(A, C) = L(A, U, C)$; we will later see
that this containment is strict.

The arguments used by \KaR~\cite{klaedtke-ruess02} apply equally well to
\KAPA and \KDAPA, showing:\squeezetopsep
\begin{proposition}\label{prop:clotapa}%
  $\LKAPA$ and $\LKDAPA$ are effectively closed under union, intersection and
  inverse morphisms.  Moreover, $\LKAPA$ is closed under concatenation and
  nonerasing morphisms, and $\LKDAPA$ is closed under complement.
\end{proposition}

We now show these models over $\bbn$ are at least as powerful as over $\bbq$.
First, we need the following technical lemma:
\begin{lemma}\label{lem:lin}
  For any \KAPA (resp. \KDAPA) there exists a \KAPA (resp. \KDAPA) where the
  functions associated with the transitions are linear, except for some
  transitions which can be taken only as the first transition of a nonempty
  run.
\end{lemma}
\begin{proof}[Proof (sketch)]
  Let $(A, U, C)$ be a \KAPA of dimension~$d$, where the transition set of
  $A$ is $\delta = \{t_1, \ldots, t_{|\delta|}\}$, and write $U(t_i) = (M_i,
  \bar{v_i})$.  Let $A'$ be a copy of $A$ in which a fresh state $q$ is
  added, set to be the initial state, with the same outgoing transitions as
  the initial state of $A$ and no incoming transition.  Let $t'_1, \ldots,
  t'_k$ be the new transitions in $A'$, and order $\delta$ such that $t_1,
  \ldots, t_k$ are the corresponding transitions leaving the initial state of
  $A$.  Now define $U'$, for $\bar{x}, \bar{y_1}, \ldots, \bar{y_{|\delta|}}
  \in \bbk^d$, by $U'(t'_i)\colon (\bar{x}, \bar{y_1}, \ldots,
  \bar{y_{|\delta|}}) \mapsto (\bar{v_i}, \bar{v_1}, \ldots,
  \bar{v_{|\delta|}})$, and $(U'(t_i)\colon (\bar{x}, \bar{y_1}, \ldots,
  \bar{y_{|\delta|}}) \mapsto (M_i.\bar{x} + \bar{y_i}, \bar{y_1}, \ldots,
  \bar{y_{|\delta|}})$.  Finally define $C' = C.\bbk^{d \times |\delta|}$.
  Then $L(A', U', C') = L(A, U, C)$, and $A'$ is deterministic if $A$ is.
  Moreover, the only nonlinear functions given by $U'$ are for the outgoing
  transitions of the initial state of $A'$, a state no run can return to.
\end{proof}

\begin{proposition}\label{prop:qnapa}
  $\LQDAPA \subseteq \LNDAPA$ and $\LQAPA \subseteq \LNAPA$.
\end{proposition}
\begin{proof}
  We first recall that a set $C \subseteq \bbq^d$ is $\bbq$-definable iff it
  is a finite union of sets of the form:
  $$\{\bar{x} \st f_1(\bar{x}) = \cdots = f_p(\bar{x}) = 0
  \land g_1(\bar{x}) > 0 \land \cdots \land g_q(\bar{x}) > 0\},$$%
  where $f_1, \ldots, f_p, g_1, \ldots, g_q\colon \bbq^d \to \bbq$ are affine
  functions (see, e.g., \cite{vandendries98}).  Let $(A, U, C)$ be a \QAPA of
  dimension $d$; by Lemma~\ref{lem:lin}, we may suppose that the functions
  associated with the transitions are linear, except for the transitions that
  may begin a run.  We suppose $C$ is a single set of the kind previously
  described; this is no loss of generality as $\LKAPA$ and $\LKDAPA$ are
  closed under union.  So let $C$ be described by functions $f_i$ and $g_i$
  as above, and suppose $d = p + q$ (we add constant $0$ functions to the
  $f_i$'s or $0$'s to the vectors of $C$ in order to do that).  Define
  $f:\bbq^d \to \bbq^d$ by $f(\bar{x}) = (f_1(\bar{x}), \ldots, g_1(\bar{x}),
  \ldots)$; clearly, $f \in \func_d(\bbq)$.  Now let $(A, U', C')$ be the
  \QAPA of dimension $2d$, defined by $(U'(t))(\bar{x}, \bar{y}) =
  ((U(t))(\bar{x}), f(\bar{x}))$, with $t$ a transition of $A$ and $\bar{x},
  \bar{y} \in \bbq^d$; and $C' = \bbq^d . \{0\}^p . (\bbq^+)^q$.  Clearly,
  $L(A, U', C') = L(A, U, C)$.  We then define $U''$ by $U''(t) = c \times
  U'(t)$ where $c$ is the maximum denominator in the reduced fractions
  appearing in the matrix and vector of $U'(t)$.  Thus, the functions given by
  $U''$ are from $\bbz^{2d}$ to $\bbz^{2d}$.  Moreover, for any $\pi \in
  \apaths(A)$, $(U''(\pi))(\bar{0}) = k \times (U'(\pi))(\bar{0})$, for some
  $k \in \bbn^+$ depending only on $\pi$.  Thus, defining $C'' =
  \bbz^d.\{0\}^p.(\bbz^+)^q$, we have $L(A, U'', C'') = L(A, U', C')$.
  Finally, the negative numbers can be circumvented by doubling the dimension
  of the matrices and keeping track of the negative and the positive
  contributions separately until the final tests for zero, which become tests
  that negative contribution equals (or is strictly lesser than) the positive
  contribution of a number (a similar technique is used by
  \KaR~\cite{klaedtke-ruess02}).
\end{proof}

\Remark The previous proof shows that the constraint set of \QAPA can be
simulated \emph{within} the automaton, and is thus of a lesser use.

We now give a large class of languages belonging to $\LQAPA$.  Define
$\Mi(L)$ as the smallest semiAFL containing $L$ and closed under
intersection; that is, $\Mi(L)$ is the smallest class of languages containing
$L$ and closed under nonerasing and inverse morphism, intersection with a
regular set, union, intersection, and concatenation.  With $\PAL = \{w\#w^R
\st w\in \{a,b\}^+\}$:
\begin{proposition}\label{prop:aflpal}
  $\Mi(\PAL) \subseteq \LQAPA$.
\end{proposition}
\begin{proof}
  We sketch a \QDAPA for $\PAL$.  The automaton starts by reading a single
  letter, if it is an $a$ it initializes its counters to $(2, 1)$, otherwise,
  it initializes them to $(2, 0)$.  Now for each letter read, if it is an
  $a$, it applies the function $(p, v) \mapsto (2p, v + p)$, and $(p, v)
  \mapsto (2p, v)$ if it is a $b$.  Upon reaching the $\#$ sign, functions
  associated to $a$ and $b$ change: when reading an $a$, the automaton
  applies $(p, v) \mapsto (p/2,v - p/2)$, otherwise it applies $(p, v)
  \mapsto (p/2, v)$.  Clearly, a word is in $\PAL$ iff it is of the form
  $\{a, b\}^+\#\{a, b\}^+$ and the final state of the counters is $(1, 0)$.
  The closure properties are implied by those of $\LQAPA$
  (Proposition~\ref{prop:clotapa}).
\end{proof}
The class $\Mi(\PAL)$ contains a wide range of languages.  First, the closure
of $\PAL$ under nonerasing and inverse morphism and intersection with
regular sets is the class of \emph{linear languages}
(e.g.,~\cite{brandenburg81}).  In turn, adding closure under intersection
permits to express the languages of nondeterministic multipushdown automata
where in every computation, each pushdown store makes a bounded number of
reversals (that is, going from pushing to
popping)~\cite{book-nivat-paterson74}; in particular, if there is only one
such pushdown store, this corresponds to the \emph{ultralinear
  languages}~\cite{ginsburg-spanier66}.  Further, as $\Mi(\COPY) \subsetneq
\Mi(\PAL)$ (e.g.,~\cite{brandenburg81}) this implies that $\COPY \in \LQAPA$.

Next, we note that \KAPA express only context-sensitive languages ($\CSL$):
\squeezetopsep
\begin{proposition}
  $\LNAPA \subseteq \CSL$.
\end{proposition}
\begin{proof}
  Let $(A, U, C)$ be an \NAPA of dimension $d$, we show that $L(A, U, C) \in
  \NSPACE[n]$ (which is equal to $\CSL$ \cite{kuroda64}).  Let $A = (Q,
  \Sigma, \delta, q_0, F)$, and $w = w_1\cdots w_n \in \Sigma^*$.  First,
  initialize $\bar{v} \leftarrow \bar{0}$ and $q \leftarrow q_0$.  Iterate
  through the letters of $w$: on the $i$-th letter, choose
  nondeterministically a transition $t$ from $q$ labeled with $w_i$.  Update
  $\bar{v}$ by setting $\bar{v} \leftarrow (U(t))(\bar{v})$ and $q$ with $q
  \leftarrow \dest(t)$.  Upon reaching the last letter of $w$, accept $w$ iff
  $q \in F$ and $\bar{v} \in C$.

  We now bound the value of $\bar{v}$.  Let $c$ be the greatest value
  appearing in any of the matrices or vectors in $U(t)$, for any $t$.  For a
  given $\bar{v}$, let $\max \bar{v}$ be $\max \{v_1, \ldots, v_d\}$.  Then
  for any~$t$, $((U(t))(\bar{v}))_i \leq d \times (c \times \max \bar{v}) +
  c$.  Let $\pi$ be a path, we then have that $((U(\pi))(\bar{0}))_i \leq
  (c(d+1))^{n-1}c$, thus the size of $\bar{v}$ at the end of the algorithm is
  in $O(n)$.  Now note that, as $C$ is semilinear, the language of the binary
  encoding of its elements is regular~\cite{wolper-boigelot95}, and thus,
  checking $\bar{v} \in C$ can be done efficiently.  Hence the given
  algorithm is indeed in $\NSPACE[n]$.
\end{proof}

We now note that the power of \KDAPA does not owe to their capabilities as
automata:\squeezetopsep%
\begin{proposition}\label{prop:dapa}%
  Let $\Sigma$ be an alphabet.  There exists a two-state automaton $A_\Sigma$
  such that for any \KDAPA over $\Sigma$, there exists a \KDAPA accepting the
  same language whose underlying automaton is $A_\Sigma$.
\end{proposition}
\begin{proof}
  Let $(A, U, C)$ be a \KDAPA of dimension $d$ where $A = (Q, \Sigma, \delta,
  q_0, F)$, with $Q = \{1, \ldots, k\}$ and $\Sigma = \{a_1, \ldots, a_m\}$.
  Let $N = k(d+1)$, we show that there exist $f_{a_1}, \ldots, f_{a_m} \in
  \func^\bbk_N$, a $\bbk$-definable set $G \subseteq \bbk^N$ and $\bar{o} \in
  \bbk^N$ such that:
  \begin{align}\label{eqn:dapa}
    w = \ell_1\cdots \ell_{|w|} \in L(A, U, C) \quad \Leftrightarrow \quad
    f_{\ell_{|w|}} \circ \cdots \circ f_{\ell_1}(\bar{o}) \in G.
  \end{align}

  Our goal is to represent the state in which the \KDAPA is with a vector of
  size $N$.  This vector is composed of $k$ smaller vectors of size $(d+1)$.
  On taking a path $\pi$ in $A$, let $q = \dest(\pi)$ and $\bar{v} =
  (U(\pi))(0^d)$; then $q$ and $\bar{v}$ describe the current configuration
  of the \KDAPA.  Thus we define, for any $q \in Q$ and $\bar{v} \in \bbk^d$:
  \def\vecc{\mathsf{Vec}}
  $\vecc(q, \bar{v}) = (
  0^{d+1}
  \,\,\cdots\,\,\,
  0^{d+1}
  \underbrace{1 \quad \bar{v}}_{q\text{-th subvector}}
  0^{d+1}
  \,\,\cdots\,\,\,
  0^{d+1}
  )$.%

  Now, for $t \in \delta$, let $M_t$ and $\bar{b}_t$ be such that $U(t) =
  (M_t, \bar{b}_t)$.  For the purpose of describing the matrix $U_a$ below,
  when $t \not\in \delta$ we let $M_t$ stand for the all-zero matrix of
  dimension $d \times d$ and $\bar{b}_t$ be the all-zero vector of dimension
  $d$.  Let $\chi$ be the characteristic function of $\delta$.  For $a \in
  \Sigma$, define:%
  \def\tallcell{\rule{0pt}{15mm}} \def\smallcell{\rule{0pt}{5mm}} \def\zzz{0
    \cdots 0}
  $$\begin{disarray}{rl}
    U_a = & \left(\begin{array}{m{0pt}@{}c|c|c|c|c}
        \smallcell
        &\chi((1,a,1))& \zzz & \cdots & \chi((k,a,1)) & \zzz \\\hline
        \tallcell
        & \bar{b}_{(1,a,1)} &M_{(1,a,1)} & \cdots &  \bar{b}_{(k,a,1)} & M_{(k,a,1)}\\\hline
        &\vdots & \vdots & \ddots & \vdots & \vdots\\\hline
        \smallcell
        & \chi((1,a,k)) &\zzz& \cdots & \chi((k,a,k))& \zzz \\\hline
        \tallcell
        & \bar{b}_{(1,a,k)} & M_{(1,a,k)} & \cdots & \bar{b}_{(k,a,k)} & M_{(k,a,k)}
      \end{array}\right)
  \end{disarray}$$

  The matrix $U_a$ is such that for $(p, a, q) \in \delta$ and $\bar{v} \in
  \bbk^d$, $U_a.\vecc(p, \bar{v}) = \vecc(q, M_{(p, a, q)}.\bar{v} +
  \bar{b}_{(p,a,q)})$.  In other words, $U_a$ computes the transition
  function and, according to the current state, applies the right affine
  function.  More generally, for a path $\pi$ in $A$ starting at $q_0$ and
  labeled by $w = \ell_1\cdots \ell_{|w|}$, we have $U_{\ell_{|w|}}\cdots
  U_{\ell_1}.\vecc(q_0, 0^d) = \vecc(\dest(\pi), (U(\pi))(0^d))$, where $0^d$
  is the all-zero vector of dimension~$d$.  We then let $G$ be the
  $\bbk$-definable set which contains $\vecc(q, \bar{v})$ iff $q$ is final
  and $\bar{v} \in C$: $G = \bigcup_{i \in F}\bigcup_{\bar{v} \in C} \vecc(i,
  \bar{v})$.

  Now let $f_{a_i} \in \func^\bbk_N$ be defined as $(U_{a_i}, 0^N)$ and let
  $\bar{o} = \vecc(q_0, 0^d)$.  Then we have precisely
  Equation~(\ref{eqn:dapa}).  Now let $A'$ be the automaton $(\{r, s\},
  \Sigma, \delta', r, \{r, s\})$ defined by $\delta' = \{r, s\} \times \Sigma
  \times \{s\}$.  Define $U'\colon \delta'^* \to \func^\bbk_N$ by:\pagebreak
  $$U'((q, a_i, q'))(\bar{x}) = \left\{\begin{disarray}{ll}
      U_{a_i}(\vecc(q_0,0^d)) & \text{if } q = r \land q' = s,\\
      U_{a_i}.\bar{x} & \text{otherwise, i.e., if } q = q' = s.
    \end{disarray}\right.$$ 
  
  Finally, a special case should be added for the empty word: We let $C' = G$
  if $\eps \not\in L(A,U,C)$ and $C' = G \cup \{0^N\}$ otherwise.  We have
  that $(A', U', C')$ is a \KDAPA where $A'$ has only two states, and it is
  of the same language as $(A, U, C)$.  Finally, note that we need two
  states, and not one, because \KAPA use $\bar{0}$ as the starting value for
  their registers but $\bar{o}$ is needed here.
\end{proof}


We now give some negative properties of APA; our main tool is the following
lemma:
\begin{lemma}\label{lem:re}
  Let $L$ be a Turing-recognizable language.  Then there exist effectively
  $L_1, L_2 \in \LQDAPA$, and a morphism $h$ such that $L = h(L_1 \cap L_2)$.
\end{lemma}
\begin{proof}
  This follows closely~\cite[Theorem 1]{baker-book74}, thus we only sketch
  the proof.  Let $M$ be a one-tape Turing machine, and suppose w.l.o.g.\
  that $M$ makes an odd number of steps on any accepting computation and that
  $M$ only halts on accepting computation.  Let $L_1$ be the set of strings
  \begin{align}
    \ID_0\#\ID_2\#\cdots\#\ID_{2k}\$(\ID_{2k+1})^R\#\cdots\#(\ID_3)^R\#(\ID_1)^R\label{eqn:conf}
  \end{align}
  such that the $\ID_i$'s are instantaneous descriptions of configurations of
  $M$, $\ID_0$ is an initial configuration, $\ID_{2k+1}$ is an accepting
  configuration, and for all $i$, $\ID_{2i+1}$ is the configuration which
  would be reached in one step from configuration $\ID_{2i}$.  Similarly,
  $L_2$ is the same as $L_1$ but checks that $\ID_{2i}$ is the successor of
  $\ID_{2i-1}$.  These languages are in $\LQDAPA$, using a technique similar
  to Proposition~\ref{prop:aflpal}.  Thus $L_1 \cap L_2$ is a language of
  $\LQDAPA$ which encodes the strings of the type of~\ref{eqn:conf} such that
  the $\ID_i$'s encode an accepting computation of $M$.  Now if each string
  $\ID_i$, $i > 0$, is over an alphabet which is disjoint from the alphabet
  which encodes the initial instantaneous description, then the morphism $h$
  which erases all of the symbols in a string of $L_1 \cap L_2$ except those
  representing the input is such that $L(M) = h(L_1 \cap L_2)$.
\end{proof}

\begin{corollary}
  Neither $\LKAPA$ nor $\LKDAPA$ is closed under morphisms.  
\end{corollary}

\begin{corollary}
  The emptiness problem is undecidable for DetAPA.
\end{corollary}
\begin{proof}
  Let $L \subseteq \Sigma^*$ be a Turing-recognizable language, and $x \in
  \Sigma^*$.  Let $L_1, L_2, h$ be given by Lemma~\ref{lem:re} for $L$.  Then
  $x \in L$ iff $L_1 \cap L_2 \cap h^{-1}(x)$ is nonempty, the latter being
  in $\LQDAPA$.
\end{proof}

Recall that $\LKAPA$ is closed under concatenation.  The previous property and
the fact that a language $L$ is empty iff $L\cdot\Sigma^*$ is finite
implies:\squeezetopsep
\begin{corollary}%
  Finiteness is undecidable for \KAPA.
\end{corollary}



\section{Parikh automata on letters}\label{sec:restr}

The PA \emph{on letters} requires that the ``weight'' of a transition only
depend on the input letter from $\Sigma$ triggering the transition.  In a way
similar to the CA characterization of PA, we characterize PA on letters
solely in terms of automata on $\Sigma$ and semilinear sets.  This model
helps us in proving a standard lemma in language theory, in the context of
PA.

\begin{definition}[Parikh automaton on letters]%
  A \emph{Parikh automaton on letters (LPA)} is a PA $(A, C)$ where whenever
  $(a, \bar{v_1})$ and $(a, \bar{v_2})$ are labels of some transitions in
  $A$, then $\bar{v_1} = \bar{v_2}$.  We write $\LLPA$ (resp. $\LDLPA$) for
  the class of languages recognized by LPA (resp. LPA which are DetPA).
\end{definition}
Now let $(A, C)$ be a LPA.  We may determinize $A$ in the standard way and,
although this is not the case with a PA, the resulting LPA is deterministic,
thus:\squeezetopsep
\begin{proposition}%
  $\LLPA = \LDLPA$.
\end{proposition}


For $R \subseteq \Sigma^*$ and $C \subseteq \bbn^{|\Sigma|}$, define $R\restr
C = \{w \in R \st \pkh(w) \in C\}$.  Then:\squeezetopsep
\begin{proposition}%
  Let $L \subseteq \Sigma^*$ be a language.  The following are equivalent:
  \begin{compactenum}[(i)]
  \item $L \in \LLPA$;
  \item There exist a regular language $R \subseteq \Sigma^*$ and a
    semilinear set $C \subseteq \bbn^{|\Sigma|}$ such that $R \restr C = L$.
  \end{compactenum}
\end{proposition}

\squeezeafterenv The following property will be our central tool for showing
nonclosure results:\squeezetopsep
\begin{lemma}\label{lem:lpareg}%
  Let $L \in \LLPA$.  For any regular language $E$:%
  $$L \cap E \text{ is not regular} \quad\Rightarrow\quad (\exists w \in E)[c(w) \cap
  L = \emptyset].$$
\end{lemma}
\begin{proof}
  Let $R \subseteq \Sigma^*$ be a regular language and $C \subseteq
  \bbn^{|\Sigma|}$ be a semilinear set.  Define $L = R \restr C$.  Let $E$ be
  a regular language such that $L \cap E$ is not regular.  As $L \subseteq
  R$, we have $(L \cap E) \subseteq (R \cap E)$.  The left hand side being
  non regular, those two sets differ.  Thus, let $w \in (R \cap E)$ such that
  $w \not\in L \cap E$, we have $w \not\in L$.  Hence, $w \in (R \setminus
  L)$, which implies that $\pkh(w) \not\in C$, and in turn, $c(w) \cap L =
  \emptyset$.
\end{proof}

\Remark Lemma~\ref{lem:lpareg} holds with, e.g., ``context-free'' in lieu of
``regular'', but the version given will suffice for our purposes.

\mbox{}
\begin{proposition}
  \begin{compactenum}[(1)]
  \item $\LLPA$ is not closed under union, complement, squaring, nonerasing
    morphisms;
  \item $\LLPA$ is closed under intersection, inverse morphisms, commutative
    closure.
  \end{compactenum}
\end{proposition}
\begin{proof}
  \emph{(1).}\quad (Union.) Let $L_1 = \{w \in \{a,b\}^* \st |w|_a=|w|_b\}$
  and $L_2 = b(a\cup b)^*$ be two languages of LPA.  Suppose $L = L_1 \cup
  L_2 \in \LLPA$.  Let $E$ be the regular language $(a^+b^+)$.
  By the pumping lemma, $L \cap E$ is not regular, thus
  Lemma~\ref{lem:lpareg} states there exists $w \in E$ such that $c(w) \cap L
  = \emptyset$.  But $u = b^{|w|_b}a^{|w|_a} \in c(w)$ and $u \in L$, a
  contradiction.\\
  (Complement.) Note that $L$ is the complement in $\{a,b\}^*$ of $\{a^mb^n
  \st m > 0 \land m \neq n\}$, which is the language of a LPA.\\
  (Squaring.)  Let $L = \{a^mb^n \st m \neq n\} \in \LLPA$.  Suppose $L^2 \in
  \LLPA$, and let $E = (a^+b^+)^2$.  Again, $L \cap E$ is not regular,
  Lemma~\ref{lem:lpareg} implies there exists $w \in E$ such that $c(w) \cap
  E = \emptyset$.  But $a^{|w|_a}b^0a^0b^{|w|_b} \in c(w) \cap L$, a
  contradiction.\\
  (Nonerasing morphisms.)  We simply note that $L$ is the image of the
  language $\{a_1^mb_1^na_2^rb_2^s \st m \neq n \land r \neq s\}$ by the
  morphism $h(a_i) = a, h(b_i) = b$.

  \emph{(2).}\quad The proofs for the first two properties follow the usual
  proofs for finite automata.  Closure under the commutative closure operator
  follows from the proof of Proposition~\ref{prop:clotpa}.
\end{proof}

Finally, we use LPA to show the following property, which has a standard form
known to be true for regular~\cite{latteux78} and context-free
languages~\cite{blattner-latteux81} (the latter recently reworked
in~\cite{ganty-majumdar-monmege10}).  This property is sometimes called
\emph{Parikh-boundedness}:
\begin{proposition}
  For any $L \in \LPA$, there exists a bounded language $L' \in \LPA$ such
  that $L' \subseteq L$ and $\pkh(L) = \pkh(L')$.
\end{proposition}
\begin{proof}
  Let $(A, C)$ be a constrained automaton, where $\delta$ is the transition
  set of $A$.  Let $R \subseteq \delta^*$ and $D \subseteq \bbn^{|\delta|}$
  be such that $\mu(R \restr D) = L(A, C)$.  As mentioned, we can find a
  bounded regular language $R' \subseteq R$ such that $\pkh(R') = \pkh(R)$.
  In particular, $\pkh(R' \restr D) = \pkh(R \restr D)$.  Closure under
  morphism of $\LPA$ implies that $L = \mu(R' \restr D)$ is a bounded
  language of $\LPA$ included in $L(A, C)$.  Moreover, $\pkh(L(A, C)) =
  \pkh(\mu(R \restr D))$, and thus, equals $\pkh(L)$.
\end{proof}


\section{Conclusion}\label{sec:conclusion}

The following table summarizes the current state of knowledge concerning the
PA and its variants studied here; a class contains the class below it, and a
language witnessing the separation is attached to the top class when we know
this containment to be strict.

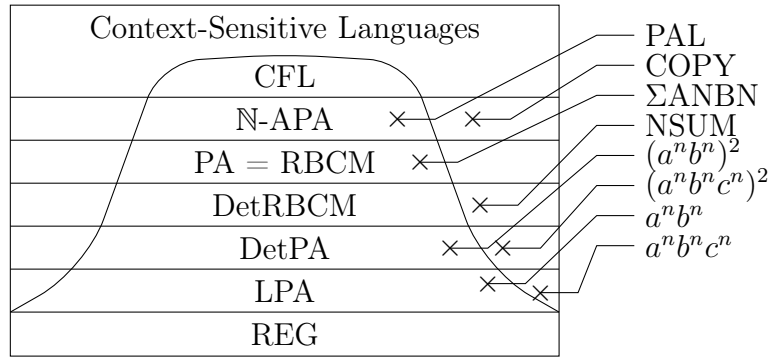
\begin{center}
  \begin{tikzpicture}
    \matrix (m) [matrix of nodes, column sep=0mm, inner sep=0pt,
      text centered,nodes={text width=7cm, inner sep=4pt}] {
      |(CSL)|Context-Sensitive Languages\\
      |(CFL)|CFL\\
      |(APA)|$\bbn$-APA\\
      |(PA)|PA = RBCM\\
      |(DRBCM)|DetRBCM\\
      |(DPA)|DetPA\\
      |(LPA)|LPA\\
      |(REG)|REG\\
    };
    \foreach \c in {REG, LPA, DPA, DRBCM, PA, APA}
       \draw (\c.north east) -- (\c.north west);
  
    \draw [rounded corners=20]
      (REG.north west) -- ($(LPA.north west) + (1, 0)$)
      [rounded corners=15] -- ($(CFL.north west) + (2, -0.1)$)  -- (CFL.north)
      -- ($(CFL.north east) + (-2, -0.1)$) [rounded corners=20] -- ($(LPA.north east) - (1, 0)$)
      -- (REG.north east);

    \draw (m.south west) rectangle (m.north east);


    \useasboundingbox (m.north east) rectangle (m.south west);
  
    \draw[node distance=0.4cm]
          node (pal) [xshift=1cm,yshift=-0.4cm,right] at (m.north east) {PAL}
          node (palx) [xshift=1.5cm] at (APA) {$\times$}
          node (copy) [below=of pal.west,right] {COPY}
          node (copyx) [xshift=2.5cm] at (APA) {$\times$}
          node (sanbn) [below=of copy.west,right] {$\Sigma$ANBN}
          node (sanbnx) [xshift=1.8cm] at (PA) {$\times$}
          node (nsum) [below=of sanbn.west,right] {NSUM}
          node (nsumx) [xshift=2.6cm] at (DRBCM) {$\times$}
          node (anbn2) [below=of nsum.west,right] {$(a^nb^n)^2$}
          node (anbn2x) [xshift=2.2cm] at (DPA) {$\times$}
          node (anbncn2) [below=of anbn2.west,right] {$(a^nb^nc^n)^2$}
          node (anbncn2x) [xshift=2.9cm] at (DPA) {$\times$}
          node (anbn) [below=of anbncn2.west,right] {$a^nb^n$}
          node (anbnx) [xshift=2.7cm,yshift=.1cm] at (LPA) {$\times$}
          node (anbncn) [below=of anbn.west,right] {$a^nb^nc^n$}
          node (anbncnx) [xshift=3.4cm,yshift=-.02cm] at (LPA) {$\times$};

    \foreach \n/\nx in {pal/palx,copy/copyx,sanbn/sanbnx,nsum/nsumx,anbn2/anbn2x,anbncn2/anbncn2x,anbn/anbnx,anbncn/anbncnx}
      \draw [to path={-- ++(-0.5, 0) -- ($(\tikztotarget) + (0.5, 0)$) -- (\tikztotarget)}]
        (\n.west) edge (\nx.center);
  \end{tikzpicture}
\end{center}
\squeezeafterenv

An intriguing question is whether there are context-free or context-sensitive
languages outside $\LNAPA$.  How difficult is that question?  How about
$\LNDAPA$?  We have been unable to locate the latter class meaningfully.  In
particular, can $\LNDAPA$ be separated from $\LNAPA$?

The following summarizes the known closure and decidability properties for PA
variants, and proposes open questions:

\def\lead#1{\multicolumn{1}{|c|}{#1}}
\begin{center}
  \let\bf=\relax
  \setlength{\tabcolsep}{5pt}
  \addtolength{\extrarowheight}{2pt}
  \begin{tabular}{c|c|c|c|c|c|c|c|c||c|c|c|c|c|}
    \cline{2-14}
    & $\cup$  & $\cap$ & $\cdot$ & $\bar{\phantom{\cdot\;}}$ & $h$  & $h^{-1}$ & $c$ & ${}^*$ &
    $\emptyset$   & $\Sigma^*$   & fin.  & $\subseteq$ & reg.\\
    \hline
    \lead{LPA}  & N & Y & N & N & N & Y & Y & N & D & D & D & D & ?\\
    \hline
    \lead{DetPA}  & Y & Y & N & Y & N & Y & Y & N & D & D & D & D & ?\\
    \hline
    \lead{PA}   & Y & Y & Y & N & Y & Y & Y & N & D & U & U & U & U\\
    \hline
    \lead{DetAPA} & Y & Y & ? & Y & N & Y & ? & ? & U & U & ? & U & ?\\
    \hline
    \lead{APA}  & Y & Y & Y & ? & N & Y & ? & ? & U & U & U & U & U\\
    \hline
  \end{tabular}
\end{center}

Several questions thus remain open concerning the poorly understood (and
possibly overly powerful) affine PA model.  But surely we expect testing a
LPA or a DetPA for regularity to be decidable.  How can regularity be tested
for these models?  One avenue for future research towards this goal might be
characterizing $\LDPA$ along the lines of algebraic automata theory.

\emph{Acknowledgments.} The first author thanks L.~Beaudou, M.~Kaplan, and
A.~Lema\^itre.

\small
\let\oldthebibliography=\thebibliography
\let\endoldthebibliography=\endthebibliography
\renewenvironment{thebibliography}[1]{%
  \begin{oldthebibliography}{#1}%
    \setlength{\parskip}{0ex}%
    \setlength{\itemsep}{0ex}%
  }%
  {%
  \end{oldthebibliography}%
}

\bibliography{ncma11}


\end{document}